\documentclass[submission,copyright,creativecommons]{eptcs}
\providecommand{\event}{NCMA 2022} 

\usepackage{iftex}

\usepackage{amsthm,amsfonts}
\usepackage{mathtools}
\usepackage{enumitem}
\def\titlerunning{Conclusive Tree-Controlled Grammars}
\def\authorrunning{D. Klobu\v{c}n\'{i}kov\'{a}, Z. K\v{r}ivka \& A. Meduna}
\usepackage{mathdots}
\usepackage[linguistics]{forest}

\newcommand{\abbba}[0]{\ensuremath{\bar{A}\bar{B}\bar{B}\bar{B}\bar{A}}}

\newcommand{\nhphantom}[1]{\sbox0{#1}\hspace{-\the\wd0}}

\usetikzlibrary{decorations.pathreplacing}
\newtheorem{theorem}{Theorem}[section]
\newtheorem{lemma}[theorem]{Lemma}

\newtheorem{corollary}[theorem]{Corollary}

\newtheorem{example}[theorem]{Example}
\newtheorem{definition}[theorem]{Definition}

\theoremstyle{remark}
\newtheorem*{basicidea}{Basic Idea}
\newtheorem*{construction}{Construction}

\ifpdf
  \usepackage{underscore}         
  \usepackage[T1]{fontenc}        
\else
  \usepackage{breakurl}           
\fi

\title{Conclusive Tree-Controlled Grammars}
\author{Dominika Klobu\v{c}n\'{i}kov\'{a} \qquad\qquad Zbyn\v{e}k K\v{r}ivka \qquad\qquad Alexander Meduna
\institute{Centre of Excellence IT4Innovations, Faculty of Information Technology,
Brno University of Technology\\ 
Bo\v{z}et\v{e}chova 2, 612 66  Brno, Czech Republic}
\email{iklobucnikova@fit.vut.cz\qquad\qquad krivka@fit.vut.cz \qquad\qquad meduna@fit.vut.cz}
}

\begin{document}
\maketitle

\begin{abstract}
This paper presents a new approach to regulation of grammars. It divides the~derivation trees generated by grammars into two sections---generative and~conclusive (the~conclusion). The former encompasses generation of symbols up till the moment when the lowest rightmost terminal of the derivation tree is generated, whereas the latter represents the final steps needed to successfully generate a sentence. A~control mechanism based on regulating only the conclusion is presented and subsequently applied to tree-controlled grammars, creating conclusive tree-controlled grammars. As the main result, it is shown that the ratio between depths of generative and conclusive sections does not influence the generative power. In addition, it is demonstrated that any recursively enumerable language is generated by these grammars possessing no more than seven nonterminals while the regulating language is union-free. 
\end{abstract}

\section{Introduction}
\label{intro}
Derivation trees serve as a graph representation of derivations leading to specific sentential forms. Naturally, since this notion corresponds to the~rewriting process of a~grammar in a~deterministic manner, it can be utilized to create a~mechanism used to regulate formal grammars. 

Such grammar types can be represented by tree-controlled grammars -- a~tree-con\-trol\-led grammar is defined as a~pair $(G, R)$, where $G$ is an~ordinary context-free grammar and~$R$ is a~regular control language (see~\cite{CulikII1977129}). The~language generated by~$(G, R)$ is defined by this equivalence: this language contains a~word $x$ if and only if $x \in L(G)$ and~there is a~derivation tree for $x$ in~$G$ such that $R$ contains every word obtained by concatenating the~symbols labeling the nodes in the same level for all levels of the~tree.
Although based upon context-free grammars, tree\--con\-trol\-led grammars are computationally complete---that is, they characterize the~family of the recursively enumerable languages. Considering this advantage, it comes as no surprise that formal language theory intensively investigates these grammars (see~\cite{Dasssow2008, Tua11}). 

However, a question presents itself: is it truly necessary to regulate all levels of a derivation tree by the~control language or would it be sufficient to verify only a few key moments of the derivation? A similar notion has been utilized by scattered context grammars with a single context-sensitive rule (see~\cite{MK20}) and by the Geffert normal form~(see \cite{geffert}), both of which are known to be equal to type-0 grammars, which are computationally complete.

Consequently, this paper presents the separation of derivation trees into two distinct parts---the \emph{generative part} and the~\emph{conclusive part} (\emph{conclusion} for short, see Figure~\ref{fig:derivationTree}). Intuitively, the former represents the~derivation tree starting from the start nonterminal and extending to the~deepest level in the derivation tree which still contains a terminal symbol, while the~latter is in charge of verifying---or concluding---the~derivation and~erasing the~remaining auxiliary nonterminals.

The~paper proposes a~different approach of regulating the~derivation tree compared to the~classic tree-controlled grammars: instead of regulating the~entire tree, it focuses specifically on regulation of the~subtree forest found in the~conclusion of the~derivation tree.  

This paper proves that  these grammars are computationally complete even if their control languages belong to the family of extended union-free regular languages (see~\cite{Dasssow2008}). We establish that any recursively enumerable language is generated by a seven\--non\-ter\-mi\-nal conclusive tree-controlled grammar $(G, R)$, where $R$ is an extended union-free regular language, and~in addition, $R$ uses no more than seven iterations ($\ast$) and~no more than ten concatenations.  

The~paper is organized as follows: first, an~introduction to graph theory, derivation trees and~their anatomy is given. Then, the conclusive tree-controlled grammars and~languages generated by them are introduced, followed by a~discussion on their generative power and~a comparison with type-0 grammars. At the~end of the~paper, an~overview of the~results is given and several open problems are listed.

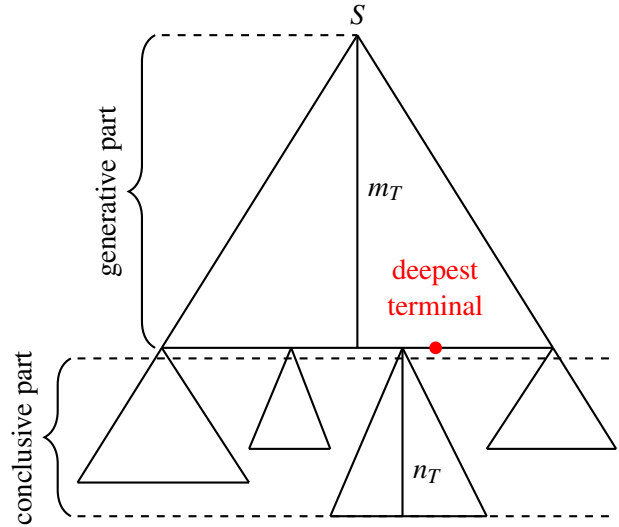
\begin{figure}[t]
\centering
\begin{tikzpicture}[thick,scale=0.8]
\coordinate (outBL) at (0, 0);
\coordinate (outBR) at (10, 0);
\coordinate (outT) at (5, 8);
\coordinate (firstW) at (3.2, 0);
\coordinate (secondW) at (4.55, 0);
\coordinate (secondmid) at (3.9, 0);
\coordinate (thirdmid) at (5.75, 0);
\coordinate (thirdW) at (7.15, 0);
    \path (outBL) -- coordinate[pos=0.35] (mainBL) (outT);
    \path (outT) -- coordinate[pos=0.65] (mainBR) (outBR);
    \draw (mainBL) -- (outT);
    \draw (mainBR) -- (outT);
    \draw (mainBR) -- (mainBL);
    \path (outBL) -- coordinate[pos=0.2] (firstBL) (mainBL);
    \draw (firstBL) -- coordinate[pos=.25] (secondBL) (mainBL);
    \draw (mainBL) -- (firstW |- firstBL);
    \draw (firstBL) -- (firstW |- firstBL);
    \draw (mainBL -| secondmid) -- (secondBL -| secondW); 
    \draw (secondBL -| firstW) -- (secondBL -| secondW);
    \draw (mainBL -| secondmid) -- (secondBL -| firstW);
    \draw (secondW) -- (thirdW);
    \draw (secondW) -- (mainBL -| thirdmid);
    \draw (mainBL -| thirdmid) -- (thirdW);
    \path (mainBR) -- coordinate[pos=.6] (fourthBR) (outBR);
    \draw (fourthBR) -- (mainBR);
    \draw (fourthBR) -- (thirdW |- fourthBR);
    \draw (thirdW |- fourthBR) -- (mainBR);
    \draw [decorate,decoration={brace,amplitude=10pt,raise=2.5pt}] (mainBL) -- (mainBL |- outT) node[xshift=-20pt,midway,rotate=90]{\normalsize generative part};
    \draw [decorate,decoration={brace,amplitude=10pt,raise=2.5pt}] (outBL -| firstBL) -- ([yshift=-5pt]firstBL |- mainBL) node[xshift=-20pt,midway,rotate=90]{\normalsize conclusive part};
    \path (mainBL) -- coordinate[pos=.7] (deepest) (mainBR);
    \filldraw[red] (deepest) circle (2.5pt) node[anchor=south, label={[align=center]deepest\\terminal}]{};
    \node[anchor=south] at (outT) {$S$};
    \draw (outT) -- (mainBL -| outT) node[anchor=west,midway]{$m_T$};
    \draw (mainBL -| thirdmid) -- (thirdmid) node[anchor=west,near end]{$n_T$};
    \draw[dashed] ([yshift=-5pt]firstBL |- mainBL) -- ([yshift=-5pt]mainBR -| fourthBR);
    \draw[dashed] (outBL -| firstBL) -- (fourthBR |- outBR);
    \draw[dashed] (mainBL |- outT) -- (outT);
    
    \end{tikzpicture}
    \caption{Separation of derivation tree, $T$, into the generative and conclusive part based on the deepest (lowest rightmost) terminal; all symbols located to its right are nonterminal. The conclusive part starts with the following level, and contains no terminals. No symbol belongs to both generative and conclusive parts; $m_T$ and $n_T$ denote the depth of the generative and conclusive part, respectively.}
    \label{fig:derivationTree}
\end{figure}

\section{Definitions}
\label{sec:def}
This paper assumes that the reader is familiar with language theory (see~\cite{introduction}).

For an~alphabet, $V$, $V^*$ represents the~free monoid generated by~$V$ under the operation of concatenation.  The unit of~$V^*$ is denoted by $\varepsilon$.  Set $V^+ =  V^* - \{\varepsilon\}$; algebraically, $V^+$ is thus the~free semigroup generated by~$V$ under the~operation of~concatenation.  For $w\in V^*$, $|w|$ denotes the~length of~$w$.  Furthermore, $\operatorname{suffix}(w)$ denotes the~set of~all suffixes of~$w$, and~$\operatorname{prefix}(w)$ denotes the set of all prefixes of $w$.  
For $w\in V^*$ and $T\subseteq V$, $\operatorname{occur}(w, T)$ denotes the number of occurrences of symbols from $T$ in $w$. For instance, $\operatorname{occur}(abdabc, \{a, d\}) = 3$. If $T  = \{a\}$, where $a\in V$, we simplify $\operatorname{occur}(w, \{a\})$ to $\operatorname{occur}(w, a)$.
For a sequence, $x = (a_1, a_2, \dots, a_n)$, where $a_i \in V$ for $1 \leq i \leq n$, $|x| = n$ denotes the length of $x$. By $\mathbb{N}$, we denote the set of all positive integers. Let $I \subset \mathbb{N}$ be a finite nonempty set. Then, $\max(I)$ denotes the~maximum of $I$.

\begin{definition}
	\label{def:ufrl}
	\emph{Union-free regular languages} (\emph{UFRL} for short) over an~alphabet $\Sigma$ are defined recursively as follows: 
	\begin{enumerate}[label=(\roman*)]
		\item{$\{\varepsilon\}$, $\emptyset$ are \emph{UFRL} over $\Sigma$;}
		\item{for every $a \in \Sigma$, $\{a\}$ is an~\emph{UFRL} over $\Sigma$;}
		\item{let $X, Y$ be \emph{UFRL}, then,}
		\begin{enumerate}[label=(\alph*)]
			\item{$XY$ is an~\emph{UFRL} (concatenation),}
			\item{$X^*$ is an~\emph{UFRL} (iteration).}
		\end{enumerate}
	\end{enumerate}
\end{definition}
\noindent The~family of~\emph{UFRL} is denoted by $\mathbf{UFREG}$.

\begin{definition}
    \label{def:eufrl}
    \emph{Extended union-free regular languages} (\emph{EUFRL} for short) over an~alphabet $\Sigma$ are defined recursively as follows: 
	\begin{enumerate}[label=(\roman*)]
		\item{\label{enum2:i} $\{\varepsilon\}$, $\emptyset$ are \emph{EUFRL} over $\Sigma$;}
		\item{\label{enum2:ii}
			for every $X \subseteq \Sigma$, $X$ is an~\emph{EUFRL} over $\Sigma$;}
		\item{\label{enum2:iii} let $X, Y$ be \emph{EUFRL}, then,}
		\begin{enumerate}[label=(\alph*)]
			\item{\label{enum2:a} $XY$ is an~\emph{EUFRL} (concatenation),}
			\item{\label{enum2:b} $X^*$ is an~\emph{EUFRL} (iteration).}
		\end{enumerate}
	\end{enumerate}
\end{definition}
\noindent The~family of~\emph{EUFRL} is denoted by $\mathbf{EUFREG}$.

Observe that $\{a, b\} \in \mathbf{EUFREG} - \mathbf{UFREG}$ but $\{a, b\}^\ast = \{a^\ast b^\ast\}^\ast \in \mathbf{UFREG}$.

A \emph{type-0 grammar} is a quadruple $G = (N, \Sigma, P, S)$, where $N$ and $\Sigma$ are the finite alphabets of nonterminals and terminals, respectively, such that $N \cap \Sigma = \emptyset$, $S \in N$ is the start nonterminal, and $P$ is the set of productions in the form of $x \to y$, where $x, y \in (N \cup \Sigma)^\ast$, $x \not\in \Sigma^\ast$. Let $V = N \cup \Sigma$. For some $p = x \to y \in P$ (a~production labeled by $p$), $\operatorname{lhs}(p)$ denotes $x$ as \emph{the left-hand side} of $p$ and~$\operatorname{rhs}(p)$ denotes $y$ as \emph{the right-hand side} of $p$. The~direct derivation relation over $V^\ast$, denoted by $\Rightarrow$, is defined as follows: $uxv \Rightarrow uyv\ [p]$ in~$G$, or simply $uxv \Rightarrow uyv$, if and only if $u, v \in V^\ast$ and~$p: x \to y \in P$. Let $\Rightarrow^n$ and $\Rightarrow^\ast$ denote the $n$th power of~$\Rightarrow$, for some $n \geq 0$, and the~reflexive-transitive closure of $\Rightarrow$, respectively.
The language generated by $G$ is denoted by $L(G)$ and defined as $L(G) = \{x : S \Rightarrow^\ast x, x \in \Sigma^\ast\}$. Two grammars are \emph{equivalent} if both generate the same language. The~family of languages generated by type-0 grammars (also known as the family of recursively enumerable languages) is denoted by $\mathbf{RE}$.

Both, \emph{right-linear grammars} and \emph{context-free grammars} are type-0 grammars, $G = (N, \Sigma,$ $P, S)$, where $N$, $\Sigma$, and $S$ have the same meaning as in the previous definition, and $P$ is the set of productions in the form of $A \to w$, where $A \in N$, $w \in \Sigma^\ast(N \cup \{\varepsilon\})$ and $A \in N$, $w \in (N \cup \Sigma)^\ast$, respectively.

\begin{definition}
	\label{def:geffert}
	A~\emph{type-0 grammar}, $G = (\{S, S', A, B\}, \Sigma, P$ $\cup$ $\{ABBBA$ $\to \varepsilon\}, S)$ is said to be in~the~third Geffert normal form if every production, $p \in P$, has one of~the~following forms:
	\begin{enumerate}[label=(\roman*)]
		\item {\label{def:geffert1}$S \to uSa$},
		\item {\label{def:geffert2}$S \to S'$},
		\item {\label{def:geffert3}$S' \to uS'v$},
		\item {\label{def:geffert4}$S' \to uv$},
	\end{enumerate}
	where $u \in \{AB, ABB\}^\ast$, $v \in \{BA, BBA\}^\ast$, and~$a \in \Sigma$. 
\end{definition}
Recall that type-0 grammars are computationally complete~\cite{Med00b}.
\begin{lemma}[Geffert normal form \cite{geffert}]
	\label{lemma:geffert}
	For every type-0 grammar, $G$, there exists an~equivalent grammar in~the~third Geffert normal form.
\end{lemma}

Let $H = (A, \rho)$ be a~directed graph, where $A$ is the~set of~\emph{nodes} and~$\rho$ is a relation on~$A$ consisting of \emph{edges}. A~sequence of nodes, $(a_0, a_1, \dots, a_n)$ for some $n \geq 1$, is a~\emph{path of length $n$} from $a_0$ to $a_n$ if $(a_{i-1}, a_i) \in \rho$ for all $1 \leq i \leq n$. If $a_0 = a_n$, the~path is a~\emph{cycle}; $H$ is \emph{acyclic} if it contains no cycles. For a~path $(a_0, a_1, \dots, a_n)$, $a_0$ is the~ancestor of $a_n$ and $a_n$ is the~descendant of $a_0$.
A~\emph{tree} is an~acyclic graph $T = (A, \rho)$ such that $A$ contains a~specific node, called the~\emph{root of $T$} and~denoted by $\operatorname{root}(T)$ and~every $a \in A - \{\operatorname{root}(T)\}$ is a descendant of $\operatorname{root}(T)$. If a node $a$ has no descendants, it is a~\emph{leaf}. Otherwise, it is an~\emph{interior node}. We shall consider $T$ to be an~ordered tree, causing every interior node $a \in A$ to have all its direct descendants $b_1\cdots b_n$ ordered from left to right. The~\emph{frontier} of $T$, denoted as $\operatorname{frontier}(T)$, is the~sequence of all leaves of $T$ ordered from left to right. The~length of the~longest path from $\operatorname{root}(T)$ in~$T$ is referred to as the~\emph{depth} of $T$, denoted by $\operatorname{depth}(T)$. For $0 \leq i \leq \operatorname{depth}(T)$, the \emph{$i$-th level} of $T$, denoted as $\operatorname{level}(T, i)$, is the~sequence of nodes of $T$ in the order defined by the ordered tree with a~path of length $i$ from $\operatorname{root}(T)$. A tree $S = (B, \nu)$ is a \emph{subtree} of $T$ if $\emptyset \subset B \subseteq A$, $\nu \subseteq \rho \cap (B \times B)$, and no node in $A - B$ is a descendant of a~node in~$B$. The tree $S$ is considered to be an~\emph{elementary subtree of $T$} if $\operatorname{depth}(S) = 1$. All trees in this paper are considered to be directed in top-down manner (so called \emph{out-degree tree}) and~ordered from left to right.

Let $G = (N, \Sigma, P, S)$ be a~context-free grammar and $A \to x \in P$ be a~production. Let $T$ be the elementary tree satisfying $\operatorname{root}(T) = A$ and~$\operatorname{frontier}(T) = x$  such that either $\operatorname{root}(T)$ has $\varepsilon$ as its only descendant, or $x = a_1\dots a_n$, $n \geq 1$ with  $\operatorname{root}(T)$ having $a_1\cdots a_n$ as its descendants; then, $T$ is considered the~\emph{production tree} representing $A \to x$. A~\emph{derivation tree} is any tree such that its root belongs to $N$ and~each of its elementary subtrees is a~production tree representing a~rule $p \in P$.
The~set of all derivation trees, $T'$, such that $\operatorname{root}(T) = S$ and~$\operatorname{frontier}(T') = x$ is denoted by $\Delta_G(x)$; by extension, the~set of all derivation trees generating a~language, $L$, is denoted by $\Delta_G(L)$. 

The~\emph{generative part} of $T$ is the~sub\-tree whose root is $S$ and the lowest level is the last level containing a terminal in $T$, whereas the~\emph{conclusive part} (or \emph{conclusion}) contains the~remaining portion of the~tree; observe that all leaves in the conclusive part are $\varepsilon$. Let $m_T$ be the~depth of the~generative part of~$T$ and~$n_T$ be the~maximum depth of subtrees from the~conclusion of~$T$ (see Figure~\ref{fig:derivationTree}). If $m_T \geq n_T$, $T$ is said to have a~\emph{short conclusion}; otherwise, it is said to have a~\emph{long conclusion}. The~\emph{depth of conclusion of $T$} is the~maximum depth of a~derivation tree located in the~conclusive part of $T$.
\begin{definition}
    Let $G = (N, \Sigma, P, S)$ be a~context-free grammar and $x \in \Sigma^\ast$ be a~string. The~set of all derivation trees, derivation trees with short conclusion and~long conclusion is, respectively, defined as
\begin{align*}
    \prescript{}{c}{\Delta_G(x)} &= \Delta_G(x),\\
    \prescript{}{sc}{\Delta_G(x)} &= \{t \in \Delta_G(x) : m_t \geq n_t\}, \text{ and}\\
    \prescript{}{lc}{\Delta_G(x)} &= \{t \in \Delta_G(x) : m_t < n_t\}, 
\end{align*}
such that $m_t, n_t$ are the~depths of the~generative and~conclusive part of the~derivation tree~$t$ of $x$.
\end{definition}

A \emph{tree-controlled grammar} (\emph{TCG} for short) is a pair $H = (G, R)$, where $G = (N, \Sigma, P, S)$ is a~context-free grammar and $R \subseteq (N \cup \Sigma)^*$ is a regular control language.
The language generated by $H$ is denoted by $L(H)$ and defined by 
\begin{align*}
L(H) = \{&x : x \in L(G), \exists t \in \Delta_G(x)\textnormal{ such that }\\
&\textrm{for all }0 \le i < \operatorname{depth}(t), \operatorname{level}(t, i) \in R\}.
\end{align*}

\begin{example}\label{ex:tcg}
Let $H =(G, R)$ with $G = (\{S, A, B, C\}, \{a, b\}, P, S)$ be a~TCG, where $P$ contains
\[\begin{array}{llll}
    1\colon\ S \rightarrow AB, & 2\colon\ A \to AA, & 3\colon\ A \to a, & 4\colon\ B \to bB,\\
    5\colon\ B \to BC, & 6\colon\ B \to \varepsilon, & 7\colon\ C \to \varepsilon,
\end{array}\]
with the regular control language $R = \{S\}$ $\cup$ $\{A\}^\ast\{b\}^\ast\{B\}^\ast\{C\}^\ast$ $\cup$ $\{a\}^\ast \{b\}^\ast \{B\}^\ast \{C\}^\ast$. In this way, we can generate $aabb$ in a successful derivation depicted by the derivation tree, $T$, in Figure~\ref{fig:derivation-tree-aabb}.
\begin{center}
\begin{figure}[t]
\centering
\begin{forest}
[$S$,name=root
    [$A$
        [$A$
            [$a$,name=borderleft]
        ]
        [$A$
            [$a$]
        ]
    ]
    [$B$
        [$b$]
        [$B$
            [$b$]
            [$B$,name=borderright
                [$B$
                    [$\varepsilon$,name=borderbottom]
                ]
                [$C$,name=bordertop
                    [$\varepsilon$]
                ]
            ]
        ]
    ]
]
\node[draw=black,dashed,thick,fit=(borderleft)(root)(borderright),inner sep=5pt] {};
\node[draw=black,dashed,thick,fit=(bordertop)(borderbottom)(bordertop -| borderbottom)] {};
\draw [decorate,decoration={amplitude=10pt}] (borderleft |- root) 
node[above,rotate=90,xshift=-50pt,yshift=15pt]{\normalsize generative part};
\draw [decorate,decoration={amplitude=10pt}] (borderbottom |- bordertop) 
node[above,rotate=90,xshift=-20pt,yshift=15pt]{\normalsize conclusive part};
\end{forest}
\caption{Derivation tree $T$ of $H$ for $aabb$  where the root is $S$}
\label{fig:derivation-tree-aabb}
\end{figure}
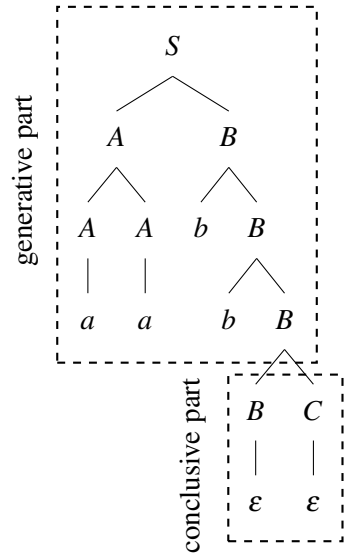
\end{center}
The rules of $G$ are applied in an arbitrary order until the rule $A \to a$ is applied on a nonterminal~$A$ on level~2 of the derivation tree, effectively forcing the other instances of $A$ to be rewritten to $a$ as well; note that the derivation tree does not necessarily reflect the order in which the individual nonterminals were rewritten. 

Observe that $L(H) = \{a^{2^i}b^j: i, j \geq 0\}$, which is a non-context-free language. The~deepest (lowest and rightmost) terminal, $b$, is located on level~3, which marks the~end of the generative part of the~derivation tree; during the conclusive part, the remaining nonterminals $B$ and $C$ are erased. Consequently, $m_T = \operatorname{max}(\{|w|-1\colon w \in \{(S, A, A, a), (S, B, b), (S, B, B, B)\}\}) = 3$, and~$n_T = \operatorname{max}(\{|w|-1\colon w \in \{(B, \varepsilon), (C, \varepsilon)\}\}) = 1$, making $T$ a derivation tree with a~short conclusion as $m_T \geq n_T$. 
\end{example} 

Now, we introduce a modification of tree-controlled grammars that utilizes a level-controlling condition only for the~conclusive part of the~derivation tree.

\begin{definition}
    \label{def:condition}
    Let $H = (G, R)$ be a TCG and $T \in \Delta_G(L(G))$. Then, the~\emph{conclusive condition} ($CC(T)$) holds if and only if for all $m_T < i \leq \operatorname{depth}(T)$, $\operatorname{level}(T, i) \in R$. 
\end{definition}

\begin{definition}
\label{def:conclang}
Let $H = (G, R)$ be a TCG. The~\emph{conclusive}, \emph{short-conclusive}, and \emph{long-conclusive language} generated by $H$ is defined by
\begin{align*}
    \prescript{}{c}{L(H)} &= \{x \in L(G) : t \in \prescript{}{c}{\Delta_G(x)} \textrm{ and } CC(t) \textrm{ holds}\}, \\
\prescript{}{sc}{L(H)} &= \{x \in L(G) : t \in \prescript{}{sc}{\Delta_G(x)} \textrm{ and } CC(t) \textrm{ holds}\}, \textit{ and}\\
\prescript{}{lc}{L(H)} &= \{x \in L(G) : t \in \prescript{}{lc}{\Delta_G(x)} \textrm{ and } CC(t) \textrm{ holds}\},
\end{align*}
respectively.
    
The~family of~conclusive, short-conclusive, and long-conclusive languages generated by tree-con\-trolled grammars are denoted by $\mathbf{CTC}$, $\mathbf{sCTC}$, and $\mathbf{lCTC}$, respectively.
    
For brevity, a TCG generating conclusive, short-conclusive, or long-conclusive language is referred to as a \emph{conclusive tree-controlled grammar} (\emph{CTCG} for short).
\end{definition}

\begin{example}
Consider the~TCG, $H = (G, R)$, from Example~\ref{ex:tcg}. Let $H' = (G, R')$ be a CTCG such that $R' = \{B\}^\ast\{C\}^\ast$. Observe that since the~regulation has been reduced only to the~conclusion of the~derivation tree, $L(H') = \{a\}^+\{b\}^\ast$.
\end{example}

\section{Results}
It has been established that tree-controlled grammars are computationally complete even if their control set is restricted to a subregular language (see~\cite{Dasssow2008, Tur12}). This section extends the~principles to conclusive tree-controlled grammars and establishes their computational completeness using an~extended union-free regular control language.

First, the~equality of languages generated by subtypes of conclusive tree-controlled grammars is established.

\begin{theorem}\label{thm:lcsc}
Let $H = (G, R)$, where $G = (N, \Sigma, P, S)$, be a~conclusive tree-controlled grammar. Then, $\prescript{}{lc}{L(H)} \subseteq \prescript{}{sc}{L(H)}$.
\end{theorem}

\begin{proof}
    Introduce a~conclusive tree-controlled grammar, $H_{\Delta} = (G_{\Delta}, R)$ such that $G_{\Delta} = (N, \Sigma, P \cup \{S \to S\}, S)$. Then, let $m$ and~$n$ be the~maximum height of generative part and~minimum height of~conclusion in~$\Delta_G(L(G_{\Delta}))$, respectively. It is apparent that by $(m-n)$ applications of the~rule $S \to S$, all sentences $x \in L(G)$ can be generated with a~short conclusion. Thus, the~theorem holds.
\end{proof}

Next, we present the~basic idea describing how to convert a~type-0 grammar $Q = (\{S, S', A, B\}, T, P \cup \{ABBBA \to \varepsilon\}, S)$ in the~third Geffert normal form to an~equivalent conclusive tree-controlled grammar $H = (G, R)$. The~idea consists in the creation of a~derivation in $G$ by context-free productions in an utterly arbitrary way, after which precisely the~substring $ABBBA$ located in the~middle of~the~sentential form is erased repeatedly during the~conclusion---that is, the~controlled part of the~derivation tree. In this way, the~correctness of the~derivation is verified.

More precisely, $G$ generates every $w \in \prescript{}{lc}{L(H)}$ by performing three consecutive phases: (I), (II), and~(III). First, by using context-free productions, the~sentential form $uSw$ is derived, where $u$ is a~string over $\{AB,$ $ABB\}^*$, and~$w$ is a~terminal string, $w \in \Sigma^\ast$. Considering $w$, phase (I) is not regulated by the~control language $R$.

Phase (II) starts with application of the~production $S \to S'$, which marks the~beginning of the~conclusion. In this phase, $G$ rewrites the~sentential form $uu'S'v'w$, where $v'$ is a~string over $\{BA, BBA\}^\ast$ representing the~nonterminal counterparts of $u, u' \in \{AB, ABB\}^\ast$. Phase (III) is entered upon replacing the~nonterminal $S'$ in the sentential form, as its presence is required to generate any additional symbols. Finally, the sub\-string $ABBBA$ found in the middle of the sentential form is repeatedly activated and erased in accordance with the control language using the following rules:
\[\begin{array}{cccc}
     A \to \bar{A}, & B \to \bar{B}, & \bar{A} \to \varepsilon, & \bar{B} \to \varepsilon.
\end{array}\]

To summarize the rewriting process, every sentence $w \in L(G, R)$ is generated by the following sequence of steps as:
\noindent
\[S \Rightarrow^\ast_{\text{(I)}} uSw \Rightarrow^\ast_{\text{(II)}} uu'S'vw \Rightarrow^\ast_{\text{(III)}} u''\abbba{}v''w \Rightarrow^\ast w,\]

\noindent
where $u'' \in \{AB, ABB\}^\ast$ and $v'' \in \{BA, BBA\}^\ast$ such that $u'' \in \operatorname{prefix}(u')$ and $v'' \in \operatorname{suffix}(v')$.

Now, using the third Geffert normal form, the~section demonstrates that for every type-0 grammar, $Q$, there exists an equivalent conclusive tree-controlled grammar, $H = (G, R)$, where $G = (N$, $\Sigma$, $P$, $S)$ and~$R$ is an~extended union-free regular control language.

\begin{theorem}
\label{th:eq}
Let $L$ be a recursively enumerable language. Then, there exists a conclusive tree-con\-trolled grammar, $H = (G, R)$, where $R \in \mathbf{EUFREG}$ such that $L = \prescript{}{c}{L(H)}$.
\end{theorem}

\begin{proof}
Let $Q = (N_Q, \Sigma, P_Q, S)$ be a type-0 grammar such that $L(Q) = L$. Without any loss of generality, assume that $Q$ conforms to the third Geffert normal form (see Definition~\ref{def:geffert}), and that $N_Q \cap \{\bar{A}, \bar{B}\} = \emptyset$. 
Let us introduce a~conclusive tree-controlled grammar, $G = (N_G, \Sigma, P_G, S)$, and extended union-free regular control language, $R$.

\def\Rii{ R_{2}}
\def\Riii{ R_{3}}
\begin{construction}
Introduce a conclusive tree-controlled grammar $H = (G, R)$ where $G = (N_Q$ $\cup$ $\{\bar{A}, \bar{B}\}$, $\Sigma$, $P_{G}$, $S)$ with $V_G = N_G \cup \Sigma$ and $V_Q = N_Q \cup \Sigma$. Let $P_G = P_{Gen} \cup P_{Act} \cup P_{Pro} \cup P_{Era}$ be constructed in the following way:
\begin{align*}
&&P_{Gen} &= \{p \in P_Q : \operatorname{occur}(\operatorname{rhs}(p), S') \geq 1 \},\\
&&P_{Act} &= \{S' \to u\abbba{}v : S' \to uv \in P_Q, u \in \{AB, ABB\}^\ast, v \in \{BA, BBA\}^\ast\}, \\
&& P_{Pro} &= \{A \to A, B \to B\}, \\
&& P_{Era} &= \{A \to \bar{A}, B \to \bar{B}, \bar{A} \to \varepsilon, \bar{B} \to \varepsilon\}.
\end{align*}

Set the partial control languages, $\Rii$ and $\Riii$, as extended union-free regular languages (see Definition~\ref{def:eufrl}) as 
\begin{align*}
    \Rii &= \{S'\} N_Q^\ast,\\
    \Riii &= \{\abbba{}\} N_Q^\ast,
\end{align*}
and the control language, $R$, an~extended union-free regular language, as

$$R = N_Q^\ast \Rii{}^\ast \Riii^\ast.$$
Observe that $\Rii \subseteq R$, and $\Riii \subseteq R$ without need of union operation, and~that $\Rii$ and $\Riii$ correspond to the~phases (II) and (III) of conclusion, respectively.

The control language, $R$, assures that both context-free phases of the re\-writing process proceed without any restrictions, and that the final phase is only finished successfully if the generated sentence belongs to $L(Q)$.
\end{construction}

\begin{basicidea}
Next, we sketch the reason why $L(Q) = \prescript{}{c}{L(H)}$. 
$H$ simulates the derivation steps of $Q$ by using a combination of context-free productions and the subregular control language. 
Phase (I) is completely contained in the generative part of $\Delta_G(x)$. 
As phase (II) generates new nonterminals and~phase (III) propagates the existing nonterminals at the beginning of a conclusive sentential form, the overall control language contains the prefix from $N_Q^\ast$. 

All context-free productions found originally in $Q$ are represented by the sets $P_{Gen}$ and $P_{Act}$.
The former set consists of all productions of form~\ref{def:geffert3} of $Q$ (see Definition \ref{def:geffert}) with the purpose of generating the terminal string and surrounding nonterminals, 
whereas the latter set represents the productions of form \ref{def:geffert4}, which may be used to effectively activate the central substring, $\abbba$, and subsequently start the erasing phase.

Similarly, the purpose of sets $P_{Pro}$ and $P_{Era}$ is to assure proper generation of the derivation tree; $P_{Pro}$ serves to propagate the corresponding nonterminals $A$ and $B$ to the lower level of the derivation tree, while $P_{Era}$ is responsible for simulation of the sole context-sensitive rule $ABBBA \to \varepsilon$ from $P_Q$ used to erase the current center of the sentential form.

Partial control language $\Rii$ simulates the use of context-free productions of form \ref{def:geffert3} so it is responsible for the generation of the nonterminal suffix needed to generate a sentence. Finally, once $S'$ has been erased from the sentential form, only the~partial control language $\Riii$ may be matched. Thanks to the properties of the~Geffert normal form, the activation and subsequent erasing process may occur only at one position in the sentential form at a time. Providing by $G$, the core substrings, $S'$ and $\abbba$, of any partial control language may only appear at most once in the entire sentential form, and their appearance is mutually exclusive; therefore, $\Rii$ and~$\Riii$ may be iterated without disrupting the consistency of the rewriting process. In this way, the equivalence of $L(Q)$ and $\prescript{}{c}{L(H)}$ is maintained.

In the following claims (\emph{If}) and~(\emph{Only if}), using proof by induction on the number of derivation steps, we prove formally that $L(Q) \subseteq \prescript{}{c}{L(H)}$ and $\prescript{}{c}{L(H)} \subseteq L(Q)$, respectively.
\end{basicidea}

Define the homo\-morphism $h$ from $V_G$ to $V_Q$ as $h(X) = X$ for all $X \in V_G - \{\bar{A}, \bar{B}\}$ and $h(X) = \varepsilon$ for $X \in \{\bar{A}, \bar{B}\}$. Furthermore, define homomorphism $h'$ from $V_G$ to $V_Q$ as $h'(X) = X$ for $X \in V_G - \{\bar{A}, \bar{B}\}$ and $h'(\bar{X}) = X$ for $\bar{X} \in \{\bar{A}, \bar{B}\}$.

\medskip

\noindent\emph{Claim (If).} $S \Rightarrow^n_Q x$ implies that $S \Rightarrow^\ast_{H} x'$, where $x = h(x')$ or $x = h'(x')$, $x \in V_Q^\ast$, $x' \in V_G^\ast$ for some $n\ge 0$ and if $x$ represents a level in conclusive part of $\Delta_G(x')$, then $x' \in R$.
\begin{proof}
\emph{Induction Basis: }Let $n = 0$. The only possible $x$ is equal to $S$, as $S \Rightarrow^0_Q S$. Similarly, $S \Rightarrow^0_{H} S$, where $S = h(S)$. \\
\emph{Induction Hypothesis: }Suppose that the claim holds for all derivations of length $j$ for some $j \geq 0$.\\
\emph{Induction Step: }Consider a derivation of the form 
\[S \Rightarrow_Q^{j+1} x.\]
Then, there also exists $y \in V_Q^\ast$ such that 
\[S \Rightarrow^j_Q y \Rightarrow_Q x\ [p].\]
By the induction hypothesis, there exists a derivation 
\[S \Rightarrow^\ast_{H} y'\text{ where }y = h(y')\text{ or }y = h'(y').\]

Considering $Q$ conforms to the third Geffert normal form, the production $p \in P_Q$ has one of the following forms (see Definition~\ref{def:geffert}):
\begin{enumerate}[leftmargin=\parindent]
	\item {\label{prf:geffert1}a context-free production in accordance with one of forms \ref{def:geffert1} through \ref{def:geffert3}},
	\item {\label{prf:geffert2}a context-free production in form \ref{def:geffert4}},
	\item {\label{prf:geffert3}the context-sensitive erasing production $ABBBA \to \varepsilon$}.
\end{enumerate}
The possibilities that may occur in $H$ based on the production form are the~following.

\begin{enumerate}
	\item The production $p$ is present in $H$; $y' \Rightarrow_{H} x'\ [p]$, where $x = h(x')$. This means that the~level the~production is applied on either is not regulated or it corresponds to the~partial control language $N_Q^\ast \Rii^\ast$, depending on whether $S$ or $S'$ is in the~middle of~$x'$.
	
	\item Let $p = S' \to uv \in P_Q$ and $p' = S' \to u\abbba{}v \in P_G$ for some $u \in \{AB, ABB\}^\ast$, $v \in \{BA, BBA\}^\ast$.
	The production $p$ serves as a transition between generating and erasing phases (II) and (III) in $Q$, which are regulated by partial control languages $N_Q^\ast \Rii^\ast$ and $N_Q^\ast \Riii^\ast$, respectively. 
\[\arraycolsep=1.4pt
  \begin{array}{lll}
 y &= \alpha S' \beta \Rightarrow_Q &\alpha uv \beta = x = h(x')\ [p]\\
 y' &= \alpha S' \beta \Rightarrow_{H} &
 \alpha u \abbba v \beta = x'\ [p'] 
\end{array}\]

	for some $\alpha \in N_Q^\ast$, and $\beta \in N_Q^\ast \Sigma^\ast$. In $H$, the~$\abbba$ substring is erased immediately after being generated together with the~generation of another activated substring $\abbba$ in the~next level of the~derivation tree. 
	\item The context-sensitive production $ABBBA \to \varepsilon \in P_Q$ is simulated by consecutive application of context-free productions, $X \to \bar{X}$, $\bar{X} \to \varepsilon$ for $X \in \{A, B\}$, to select the $ABBBA$ substring and subsequently erase it. In $Q$, this erasure is performed in one derivation step and it is sufficient to mark the ensuing $\abbba$ substring in $H$.	
	\[\arraycolsep=1.4pt
	\begin{array}{lcll}
	y\phantom{'} =& \alpha ABBBA \beta &\Rightarrow_Q &\alpha \beta = x = h(x')\\
	y' =& \alpha \abbba \beta &\Rightarrow_{H}^0 &x'
	\end{array}	\]
	for some $\alpha \in N_Q^\ast$, and $\beta \in N_Q^\ast \Sigma^\ast$.	
	
	To assure that all required nonterminals placed next to each other are affected, the level of the derivation tree is regulated by the partial control language $N_Q^\ast \Riii^\ast$. Notice that all terminals occurred in the previous levels of $\Delta_G(x')$ so we have only nonterminal strings in the conclusive levels to control by $N_Q^\ast \Riii^\ast$.

\end{enumerate}
Thus, this claim conforms to the rules of induction.
\end{proof}

\noindent\emph{Claim (Only if).} $S \Rightarrow_{H}^n x$ implies that $S \Rightarrow_Q^\ast x'$, where $x \in V_G^\ast$, $x' \in V_Q^\ast$, such that $x' = h(x)$ or $x' = h'(x)$ for some $n \geq 0$ and if $x$ represents a level in conclusive part of $\Delta_G(x)$, then $x \in R$.
\begin{proof}
\emph{Induction Basis: }$S \Rightarrow^\ast_H S = x$ implies $S \Rightarrow_Q^0 S = x'$ where $x = h(x')$.\\
\emph{Induction Hypothesis: }Suppose that the claim holds for all derivations of length $j$ for some $j \geq 0$.\\
\emph{Induction Step: }Consider a derivation of the form
\[S \Rightarrow_{H}^{j+1} x.\]
Then, there also exists $y \in V_G^\ast$, such that
\[S \Rightarrow_{H}^j y \Rightarrow_{H} x.\]
By the induction hypothesis, there exists a derivation
\[S \Rightarrow_Q^* y', \text{ where } h(y) = y' \text{ or }h'(y) = y'.\] 
According to the subsets of $P_G$ to which the~used production, $p$, belongs to, four cases in $Q$ follow.

\begin{enumerate}
	\item 
	Production $p \in P_{Gen}$ containing $S'$ in $\operatorname{rhs}(p)$.
	\begin{enumerate}
		\item \label{pgenS'} The~form of $p$ depends on the moment of the rewriting when it is applied at; it either serves as the~entry point of the~conclusion, by using the production~$p = S \to S'$, or it is used to generate nonterminals to the right of $S'$ and~$p = S' \to uS'v$, $v \in \{BA, BBA\}^\ast$. The production $S' \to uS'v$ is used on the level described by the partial control language $N_Q^\ast \Rii^\ast$.
	\end{enumerate}
	
In this case, the~production may be applied in $Q$ in a way analogous to $H$,
	\[y' \Rightarrow_Q x'\ [p], \text{ where } h(x) = x'. \]
	Considering the structure of the production, it is clear that 
	\begin{equation*}
	 \operatorname{occur}(\operatorname{lhs}(p), S') = \operatorname{occur}(\operatorname{rhs}(p), S') = 1
	\end{equation*}
	for all $p \in P_{Gen}$; thus, no extra $S'$ symbols are generated. The nonterminals that were not affected by the production, and are already present in the sentential form, are nondeterministically propagated using the productions of $P_{Pro}$; otherwise, the~application of productions of $P_{Era}$ would cause the rewriting process to halt in the future.
	
	Consequently, the corresponding partial control language, $\Rii^\ast$, is iterated precisely once while the~nonterminal $S'$ is present in the~sentential form, as each iteration must contain one of the aforementioned nonterminals.
	\item 
	Production $p \in P_{Pro}$ serves to propagate the corresponding nonterminal to the following level of the de\-ri\-va\-tion tree while not affecting the sentential form. Because of this, application of $p$ in $Q$ is equal to
	\[y' \Rightarrow_Q^0 y' = x'.\]
	Considering the production $p \in P_{Pro}$ works with nonterminals $\{A, B\} \subset N_Q$, its application is allowed at any place of the~control language, $N_Q^\ast$, $\Rii^\ast$, and~$\Riii^\ast$.
	\item Production $p \in P_{Act}$ provides the transition between phases (II) and (III) of the rewriting process, which are described by partial control languages $\Rii^\ast$ and $\Riii^\ast$, respectively. Necessarily, $p$ has the form of $S' \to u\abbba{}v$, where $u \in \{AB,$ $ABB\}^\ast$, $v \in \{BA, BBA\}^\ast$. 
	This process may be described as
	\[y = \alpha S'\beta \Rightarrow_{H} \alpha u\abbba v\beta = x \ [p],\]
	with $h(x) = h(\alpha u \abbba v \beta) = h(\alpha uv \beta)$ for some $\alpha \in N_Q^\ast$, $\beta \in N_Q^\ast \Sigma^\ast$, and
	\[y' = \alpha S'\beta \Rightarrow_Q h(\alpha uv\beta) = x' = h(x)\ [S' \to uv]. \]
	
	Subsequently, the $\abbba$ substring is removed in $H$ using the productions of $P_{Era}$. 
	\item Productions $p \in P_{Era}$ are used to repeatedly select nonterminals of the $ABBBA$ substring and erase them in an inside-out way. The initial place of erasure is determined by the location of the~nonterminal $S'$.
	\begin{enumerate}
		\item Let $p_X = X \to \bar{X}$, $X \in \{A, B\}$. Production $p_X$ nondeterministically marks the nonterminal to~be e\-rased. However, this step is only required to simulate the context-sensitive production in $H$, and therefore application of~$p_X$ does not affect the sentential form of~$Q$. Production $p_X$ may be applied in the following way:
		\begin{align*}
		y &= \alpha uXv \beta \Rightarrow_{H} \alpha u\bar{X}v \beta = x\ [p_X]
		\end{align*}
		
		with $uXv \in \{A, \bar{A}\}\{B,\bar{B}\}^3\{A, \bar{A}\}$ and $X \in \{A,B\}$, 
		and whose equivalent in~$Q$ would be as follows:
		\[y' = \alpha ABBBA \beta \Rightarrow_Q^0 \alpha ABBBA \beta = h'(x).\]
				
		\item Let $p_{\bar{X}} = \bar{X} \to \varepsilon$, $X \in \{A, B\}$. Production $p_{\bar{X}}$ erases the previously selected nonterminal to simulate the production $ABBBA \to \varepsilon \in P$. It can be applied either immediately after the application of a production from $P_{Act}$, or as the follow-up of productions from $P_{Pro}$ to generate the~next level of the~derivation tree. Considering the $\abbba$ substring only serves as an intermediary in the erasing process, it may be ignored in $Q$ completely;
		\[
		\arraycolsep=1.4pt
		\begin{array}{lll}
		y &= \alpha \bar{u} \bar{X} \bar{v} \beta \Rightarrow_{H}\alpha \bar{u}\bar{v} \beta =  x\phantom{'}\\
		y' &= \alpha\beta\Rightarrow_Q^0 x' = y'
		\end{array}
		\]
		where $\bar{X} \in \{\bar{A}, \bar{B}\}, \bar{u}, \bar{v} \in \{\bar{A}, \bar{B}\}^\ast$ such that $\bar{u}\bar{X}\bar{v} \in \{\bar{A}^{i'} \bar{B}^{j'} \bar{A}^{k'} : i' \in \{0,1\}, j' \in \{0, 1, 2, 3\}, k' \in \{0,1\}\}$, $\operatorname{occur}(\alpha, \{\bar{A}, \bar{B}\}) = \operatorname{occur}(\beta, \{\bar{A}, \bar{B}\}) = 0$ and~$h(y) = y'$.
	\end{enumerate}
	These productions may only be applied during phase (III) of the rewriting process, which is regulated by the partial control language $N_Q^\ast \Riii^\ast$. Because of this, the productions $p_X$ and~$p_{\bar{X}}$ have to be applied on all nonterminals of the $ABBBA$ and $\abbba$ substrings, respectively, as seen in Figure~\ref{fig:abbba}. Selection of any other symbols would cause the rewriting process to halt on the following level of the derivation tree.

	Because of the properties of the Geffert normal form, the sentential form will always contain at most one occurrence of the substring $ABBBA$. Therefore, the~partial control language $\Riii^\ast$ is always iterated precisely once as long as the sentential form contains some $\bar{A}$s or $\bar{B}$s. 
\end{enumerate}
\begin{figure}
	\centering
	\begin{forest}
		[$S$\vspace{3mm}
			[$\cdots$, edge label = {node [right, midway, above] {$\iddots$} }, no edge 
			[$\cdots$, no edge]]
			[$A$, no edge [$\bar{A}$]]
			[$B$, no edge [$\bar{B}$]]
			[$\bar{A}$, no edge [$\varepsilon$]]
			[$\bar{B}$,  no edge [$\varepsilon$]]
			[$\bar{B}$, edge label = {node [right, midway, above] {$\phantom{w}\vdots$} }, no edge [$\varepsilon$]]
			[$\bar{B}$,  no edge [$\varepsilon$], no edge]
			[$\bar{A}$, no edge [$\varepsilon$]]
			[$B$, no edge [$\bar{B}$]]
			[$B$, no edge [$\bar{B}$]]
			[$A$, no edge [$\bar{A}$]]
			[$\cdots$, no edge, edge label = {node [right, midway, above] {$\ddots$} }, 
			[$\cdots$, no edge]]
		]
	\end{forest}
	\caption{Derivation tree reflecting phase (III) of the rewriting process. Productions of $P_{Era}$ are applied repeatedly on the~same level to simulate application of $ABBBA \to \varepsilon$.}
	\label{fig:abbba}
\end{figure}

It is clear that $p \in P - P_{Pro}$ can only be used during a~specific phase of the~rewriting process controlled by the~corresponding expression. If this were to be violated, a sentential form not described by $R$ would arise, resulting in blocking the~derivation in~$H$.

It is important to note that iteration of partial control language $\Rii$ and $\Riii$ depends on the presence of their core substring, $S'$ and $\abbba$, in their respective order. To generate a sentence, $x$, it is necessary that TCG $H$ goes through all of the following sentential forms.

\[
\arraycolsep=1.4pt
\centering
\begin{array}{rll|l}
S & \Rightarrow_{H}^\ast &\alpha_1 S x \hfill &\ \alpha_1 \in N_Q^\ast, x \in \Sigma^\ast \\
{}& \Rightarrow_{H} &\alpha_1 S' x & {} \\
{}& \Rightarrow_{H}^\ast &\alpha_1\alpha_2 S' \beta_2 x &\ \alpha_2, \beta_2 \in N_Q^\ast \\
{}& \Rightarrow_{H}^\ast &\alpha_3 \abbba \beta_3 x \phantom{owo}&\ \alpha_3, \beta_3 \in N_Q^\ast\\
{}& \Rightarrow_{H}^\ast &x &
\end{array}
\]
This implies that the~presence of the individual core substrings is mutually exclusive. As a result, only one of the partial control languages, $\Rii$, $\Riii$, may be positively iterated at a time; in that case, it is iterated precisely once. Thus, $\prescript{}{c}{L(H)} \subseteq L(Q)$.
\end{proof}

\noindent\emph{Claim (Long-conclusiveness).} $\prescript{}{c}{L(H)}$ is also long-conclusive.

\begin{proof}
     Let $x = a_1a_2\cdots a_n$ where $a_i \in \Sigma$ for all $1 \leq i \leq n$ generated as 
     \begin{equation*}
     S \Rightarrow_{H}^\ast \alpha_1 S x \Rightarrow_{H} \alpha_1 S' x \Rightarrow_{H}^\ast \alpha_1\alpha_2 S' \beta_2 x \Rightarrow_{H}^\ast \alpha_3\abbba \beta_3 x \Rightarrow_{H}^\ast x
     \end{equation*}
     where $\alpha_i, \beta_j \in N_Q^\ast$ for all $1 \leq i \leq 3, 2 \leq j \leq 3$ and~$x \in \Sigma^\ast$. 
     The~depth of the~generative part for any $T \in \Delta_G(x)$ is $n + 1$ because of the~properties of the~Geffert normal form. Subsequently, the~conclusion has the~minimum depth of~$k + m + 1$, where $k \ge 1$ represents the~number of levels needed to generate $\beta_2$, and $m \ge 1$ is the~number of levels needed to verify the~derivation or, in other words, erase the~substrings $\abbba$ located in the~middle of the~sentential form. Since $\alpha_1$ contains at least $n$ occurrences of $A$, $m \ge n$, so $k + m + 1 > n + 1$ and the~theorem holds.
\end{proof}

\noindent\emph{Claim (Iff).} 
	$S \Rightarrow_Q^\ast x$ if and only if $S \Rightarrow_{H}^\ast x$ where $x \in \Sigma^\ast$. 

\begin{proof}
	Consider $x \in \Sigma^\ast$ in Claims (\emph{If}) and~(\emph{Only if}) of the~previous proof. Since $x = x'$ as $h$ and~$h'$ for terminal symbols is the~identity, thus this claim holds.
\end{proof}
\noindent By Claims (\emph{Iff}) and (\emph{Long-conclusiveness}), $L(Q) = \prescript{}{lc}{L(H)}$, and Theorem~\ref{th:eq} holds.
\end{proof}

\begin{corollary}\label{cor:lcsc}
 $\mathbf{CTC} = \mathbf{sCTC} = \mathbf{lCTC} = \mathbf{RE}$.
\end{corollary}
\begin{proof}
By Theorem \ref{th:eq}, $\mathbf{lCTC} = \mathbf{RE}$. Every CTCG can be simulated by a Turing machine, so $\mathbf{sCTC} \subseteq \mathbf{RE}$. By Theorem \ref{thm:lcsc}, $\mathbf{lCTC} \subseteq \mathbf{sCTC}$, so $\mathbf{sCTC} = \mathbf{RE}$.    
\end{proof}

Observe that the control language $R$ can be replaced by a~union-free language $\hat{R} = (A^\ast B^\ast S'^\ast)^*$ $(\{\abbba\}^*(A^\ast B^\ast S'^\ast)^*)^*$ while the previous proof technique still works.
\begin{lemma}
\label{lemma:lin}
	The~union-free regular language, $\hat{R}$, can be generated by a right-linear grammar, $G_{\hat{R}}$, with the~nonterminal complexity of~$1$.
\end{lemma}
\begin{proof}
Recall that $N_G = \{S, S', A, B, \bar{A}, \bar{B}\}$. Let $G_{\hat{R}} = (\{\hat{S}\}, N_G, P_{\hat{R}}, \hat{S})$ be a~right-linear grammar, where $P_{\hat{R}}$ is defined as follows:
\begin{equation*}
    P_{\hat{R}} = \{\hat{S} \to x\hat{S} : x \in \{A,B,S'\}\} \cup \{\hat{S} \to \abbba \hat{S}, \hat{S} \to \varepsilon\}.
\end{equation*}
It can easily be verified that $L(G_{\hat{R}}) = \hat{R}$, and therefore Lemma~\ref{lemma:lin} holds.
\end{proof}

It has already been shown that any recursively enumerable language can be generated by a tree-controlled grammar with a regular control language whose nonterminal complexity is equal to~7 (see~\cite{Tua11}); however, the~result uses the union operations in the control language. Our regular expression for the control language uses 5 iterations and 8 concatenations. 

\begin{theorem}
\label{lemma:complexity}
	Any recursively enumerable language, $L$, can be generated by a~conclusive tree-controlled grammar controlled by a~union-free regular language using seven nonterminals.
\end{theorem}
\begin{proof}
	Let $Q'$ be a type-0 grammar such that $L = L(Q')$ and construct $H' = (G', \hat{R})$ as a~conclusive tree-controlled grammar such that $L(Q') = L(H')$, where $G'$ is constructed in accordance with the~proof of Theorem~\ref{th:eq}. $G' = (\{S, S', A, B,$ $\bar{A}, \bar{B}\}, \Sigma_{G'}, P_{G'}, S)$; similarly, let $G_{\hat{R}}$ be a right-linear grammar constructed in accordance with the~proof of Lemma~\ref{lemma:lin}, such that $L(G_{\hat{R}}) = \hat{R}$ and $G_{\hat{R}} = (\{\hat{S}\},$ $\{S, S', A, B, \bar{A}, \bar{B}\}$, $P_{\hat{R}}, \hat{S})$. Clearly, the nonterminal complexity of grammars $G'$ and $G_{\hat{R}}$ is 6 and~1, respectively, bringing the overall nonterminal complexity of $H'$ to~7.
\end{proof}

\section{Conclusion and~Open Problems}
We conclude this paper by remarking on some of the~properties of conclusive tree-controlled grammars. Although this modification is based upon the~same principle as the~original tree-controlled grammars, its main advantage lies in the~fact that until all terminals have been generated, the~derivation tree is not regulated, and thus, the~modification offers significantly lower descriptional complexity. Observe that all families of languages generated by conclusive tree-controlled grammars are equivalent, meaning the~length of the~conclusion should not affect the~generative power in any way.

Finally, we propose four open problems regarding the~conclusive modification of tree-controlled grammars:
\begin{enumerate}
    \item Consider conclusive tree-controlled grammars with a short conclusion. What is the minimum depth of conclusion needed to maintain the~computational completeness of the~grammars or the minimum ratio of the depths of the generative and conclusive parts?
    \item What is the minimum possible nonterminal complexity of conclusive tree-controlled grammars? Can it be further restricted beyond seven nonterminals?
    \item Introduce a~modification of conclusive tree-controlled grammars whose core grammar is at most linear. Does this modification affect the~generative power of conclusive tree-controlled grammars?
    \item Study other formal grammars working in a~conclusive way, where generative and conclusive parts of the derivation tree can be distinguished. 
\end{enumerate}

\section*{Acknowledgment}
\small
This work was supported by the~Ministry of Education, Youth and~Sports of Czech Republic project ERC.CZ no. LL1908  and the~BUT grant FIT-S-20-6293.

\bibliographystyle{eptcs}
\bibliography{generic}
\end{document}